\documentclass[letterpaper, 10 pt, conference]{ieeeconf}  

\IEEEoverridecommandlockouts                              

\usepackage{url}
\usepackage{cite}
\usepackage{amsmath,amssymb,amsfonts}
\usepackage{algorithmic}
\usepackage{graphicx}
\usepackage{textcomp}
\usepackage[dvipsnames]{xcolor}
\usepackage{soul}
\usepackage{multirow}
\usepackage{tikz}
\usetikzlibrary{patterns,decorations.pathmorphing,positioning}

\def\BibTeX{{\rm B\kern-.05em{\sc i\kern-.025em b}\kern-.08em
    T\kern-.1667em\lower.7ex\hbox{E}\kern-.125emX}}

\newcommand{\vo}[1]{\boldsymbol{#1}}
\newcommand{\mo}[1]{\boldsymbol{#1}}
\newcommand{\A}{\vo{A}}
\newcommand{\B}{\vo{B}}
\newcommand{\F}{\vo{F}}
\newcommand{\x}{\vo{x}}
\newcommand{\y}{\vo{y}}
\newcommand{\z}{\vo{z}}
\newcommand{\w}{\vo{w}}
\newcommand{\n}{\vo{n}}
\newcommand{\e}{\vo{e}}
\newcommand{\dd}{\vo{d}}
\newcommand{\Lg}{\vo{L}} 

\newcommand{\nx}{N_x}
\newcommand{\ny}{N_y}
\newcommand{\nz}{N_z}

\newcommand{\nd}{N_d}

\newcommand{\Htwo}{\mathcal{H}_2}
\newcommand{\Hinf}{\mathcal{H}_{\infty}}
\newcommand{\norm}[2]{\left\lVert#1\right\rVert_{#2}}
\newcommand{\vt}[1]{\vo{\tilde{#1}}}

\newcommand{\Bd}{\vo{B}_d}

\newcommand{\Cy}{\vo{C}_y}
\newcommand{\Cz}{\vo{C}_z}
\newcommand{\Dd}{\vo{D}_d}

\newcommand{\Dn}{\vo{D}_n}

\newcommand{\Sd}{\vo{S}_d}
\newcommand{\Sn}{\vo{S}_n}

\newcommand{\zerr}{\vo{\varepsilon}}

\newcommand{\Gwz}{\vo{\mathcal{G}}_{\vo{{\tilde{w}}}\zerr}(s)}
\newtheorem{theorem}{Theorem}
\newtheorem{lemma}[theorem]{Lemma}

\newcommand{\xdot}{\dot{\vo{x}}}

\newcommand{\BDelta}{{{\Delta}}}
\newcommand{\voDelta}{{{\vo{\Delta}}}}

\newcommand{\set}[1]{\mathcal{#1}}

\newcommand{\I}[1]{\vo{I}_{#1}}

\newcommand{\X}{\vo{X}}

\newcommand{\W}{\vo{W}}
\newcommand{\Y}{\vo{Y}}
\newcommand{\R}{\vo{R}}
\newcommand{\M}{\vo{M}}
\newcommand{\N}{\vo{N}}
\newcommand{\Q}{\vo{Q}}

\newcommand{\sym}[1]{\textbf{sym}\left(#1\right)}

\renewcommand{\P}{\mo{P}} 
\newcommand{\Real}{\mathbb R}

\newcommand{\inner}[1]{\left\langle \vo{e}\phi_i\right\rangle}

\newcommand{\eqnlabel}[1]{\label{eqn:#1}}

\newcommand{\eqn}[1]{(\ref{eqn:#1})}

\newcommand{\fig}[1]{Fig. (\ref{fig:#1})}

\DeclareMathAlphabet{\mathbfsf}{\encodingdefault}{\sfdefault}{bx}{n}

\newcommand{\Z}{\vo{Z}}

\newcommand{\C}{\vo{C}}
\newcommand{\D}{\vo{D}}
\newcommand{\E}{\vo{E}}

\renewcommand{\H}{\vo{H}}

\newcommand{\domain}[1]{\set{D}}

\newcommand{\diag}{\textbf{diag}}

\title{\LARGE \bf Sparse Sensing and Optimal Precision: Robust $\Hinf$ Optimal Observer Design with Model Uncertainty}

\author{Vedang M. Deshpande$^{1}$ and Raktim Bhattacharya$^{2}$
\thanks{This work was supported by the National Science Foundation (grant
number: 1762825).}
\thanks{$^{1}$Vedang M. Deshpande is a Ph.D. student in Aerospace Engineering, Texas A\&M University, College Station, TX 77843, USA. {\tt\small vedang.deshpande@tamu.edu}}%
\thanks{$^{2}$Raktim Bhattacharya is Associate Professor in Aerospace Engineering,
Electrical \& Computer Engineering, Texas A\&M University, College Station, TX 77843, USA. {\tt\small raktim@tamu.edu}}}

\begin{document}
\maketitle
\thispagestyle{empty} 
\begin{abstract}
We present a framework which incorporates three aspects of the estimation problem, namely, sparse sensor configuration, optimal precision, and robustness in the presence of model uncertainty. The problem is formulated in the $\Hinf$ optimal observer design framework. We consider two types of uncertainties in the system, i.e.
structured affine and  unstructured uncertainties. The objective is to design an observer with a given $\Hinf$ performance index with minimal number of sensors and minimal precision values, while guaranteeing the performance for all admissible uncertainties. The problem is posed as a convex optimization problem subject to linear matrix inequalities. Numerical simulations demonstrate the application of the theoretical results presented in this work.
\end{abstract}


\section{INTRODUCTION}
The problem of sparse sensor selection typically deals with an estimator's design while using a minimal number of sensors from the available set. This is a well-studied problem as there exist a number of works by various researchers \cite{boyd2009tac, Sundaram2017automatica, pappas2016acc, das2017icssa, skelton2008jour, saraf2017acc, deshpande2021cdcLCSS, Lopez2014acc,  Wolfrum2014tac, Stefanopoulou2020battery, Jovanovic2019TAC,  Chepuri2015tac,  das2020ifac, hiramoto2000vibration}.
This problem has been formulated for both continuous and discrete-time systems in different frameworks such as Kalman filter \cite{Sundaram2017automatica, pappas2016acc, das2017icssa}, $\Htwo/\Hinf$ optimal estimation
 \cite{skelton2008jour, saraf2017acc, Lopez2014acc, Wolfrum2014tac, deshpande2021cdcLCSS, Stefanopoulou2020battery}, and in terms of Cramér–Rao bound without any restrictions on the type of an estimator \cite{Chepuri2015tac}.

It is typical for most sparse sensing formulations to assume that the sensor precisions are known and fixed \cite{boyd2009tac, Sundaram2017automatica, Chepuri2015tac, Jovanovic2019TAC, Lopez2014acc, Wolfrum2014tac, Stefanopoulou2020battery}. However, this assumption often limits the performance of the control and estimation algorithms designed for control systems. On the other hand, at the design time, it might be unclear how precise a sensor should be to achieve the pre-specified performance criterion with a plausible risk of using sensors with unnecessarily high precisions, which increases the economic cost. In \cite{skelton2008jour}, authors presented a framework in which sensor and actuator precisions are treated as variables with economic cost constraints on them. The problem is solved in a convex optimization framework with guaranteed steady-state covariance bounds. They also proposed an ad-hoc algorithm to reduce the number of sensors by iteratively removing those with the least precision. An extension of their work for uncertain plants in an $\Htwo$ formulation is presented in \cite{saraf2017acc}.

Motivated by \cite{skelton2008jour}, in our recent work \cite{deshpande2021cdcLCSS}, we presented an integrated framework for sparse sensor selection, which also minimizes the required sensor precision in the context of $\Htwo/\Hinf$ optimal observer design. This paper generalizes the framework presented in \cite{deshpande2021cdcLCSS} for systems with model uncertainties but limited to an $\Hinf$ formulation.

In related work, the authors of \cite{Stefanopoulou2020battery} considered the problem of sparse sensing for uncertain systems for a specific application of battery temperature estimation. However, they assumed that the sensor precision is known and fixed, and they solved the problem via exhaustive search, which is a combinatorial problem and does not scale well for large-scale systems.

\subsubsection*{Contribution and novelty}
In this paper, we present a theoretical framework that incorporates three aspects of the estimation problem: sparse sensor configuration, optimal precision, and robustness in model uncertainty.
In particular, we discuss the $\Hinf$-optimal observer design for uncertain systems. The objective here is threefold. First, we are interested in identifying a sparse sensor configuration. Second, we want to minimize the required sensor precision to realize the sparse configuration. And finally, the sparse observer should satisfy the specified $\Hinf$ performance criterion for all admissible uncertainties. We consider the following two classes of uncertain systems: (i) systems with structured affine uncertainty in the system matrices, and (ii) systems with unstructured uncertainty which can be expressed in the linear fractional transformation (LFT) \cite{zhouBookRobust} framework. We present results to determine sparse and robust observers for both types of uncertainties.

The organization of this paper is as follows. The sparse robust $\Hinf$ observer design problems are formulated in Section \ref{sec:prob}, and solutions to these problems are presented in Section \ref{sec:thms}. In Section \ref{sec:ex}, we show numerical simulations. Concluding remarks are discussed in Section \ref{sec:concl}.

\section{Problem Formulation} \label{sec:prob}
\subsection{Notation}
The set of real numbers is denoted by $\Real$. Bold uppercase (lowercase) letters denote matrices (column vectors). $\I{}$ and $\vo{0}$ respectively denote an identity matrix and a zero matrix of suitable dimensions.
Define $\sym{\X}:=\X+\X^T$, where $\X^T$ denotes transpose of $\X$.
 Symmetric positive (negative) definite matrices are denoted by the inequality $\X>0$ ($\X<0$). $\diag(\x)$ denotes a diagonal matrix whose diagonal elements are the vector $\x$. Similarly, $\diag\left(\X_1,\X_2,\cdots,\X_N \right)$ denotes a block diagonal matrix. All inequalities and exponents of a vector are to be interpreted elementwise.

\subsection{Systems with structured affine uncertainty}
\subsubsection{Plant}
Consider the following LTI system
\begin{subequations}
\begin{align}
    \xdot &= (\A+\BDelta\A)\x + (\Bd + \BDelta\Bd)\dd, \eqnlabel{sys_proc}\\
    \y &= \Cy\x + \Dd\dd + \Dn\Sn\n, \\
    \z &= \Cz\x,
\end{align} \eqnlabel{sys}
\end{subequations}
where, $\x\in\Real^{\nx}$ is the state vector, $\y\in\Real^{\ny}$ is the vector of measured outputs, and $\z\in\Real^{\nz}$ is the output vector we are interested in estimating.  The process noise $\dd\in\Real^{\nd}$ and the sensor noise $\vo{n} \in \Real^{\ny}$ are $\mathcal{L}_2$-norm bounded signals. The process equation \eqn{sys_proc} is assumed to be independent of the sensor noise. The matrices $\BDelta\A\in\set{A}$ and $\BDelta\Bd\in\set{B}$ denote uncertainty in the system defined as
\begin{subequations}
\begin{align}
\set{A} &:=\{\BDelta\A \ |\ \BDelta\A=\M_1\F_1\N_1, \F_1^T\F_1\leq\I{}\}, \\
\set{B} &:=\{\BDelta\B \ |\ \BDelta\B=\M_2\F_2\N_2, \F_2^T\F_2\leq\I{}\},
\end{align} \eqnlabel{def_unc_set}
\end{subequations}
where $\M_1,\N_1,\M_2,\N_2$ are known deterministic matrices.

The nominal system matrices $\A, \Bd, \Cy, \Cz, \Dd$ are known constant real matrices of appropriate dimensions. The diagonal matrix $0<\Sn\in\Real^{\ny\times\ny}$ is an unknown scaling matrix to be determined. As discussed below in Section \ref{sec:sensor_prec}, $\Sn$ is related to the precision of sensors. We also assume that the individual sensor channels are independent of each other, i.e., $\Dn = \I{}$. All other weightings or scaling matrices are assumed to be known and absorbed in the system matrices.

\subsubsection{Observer and error system}
Now, let us consider the state observer for the system \eqn{sys} given by
\begin{align}
    \dot{\hat{\x}} = \left(\A+\Lg\Cy \right)\hat{\x} - \Lg\y,  \ \ \
    \hat{\z} = \Cz\hat{\x}, \eqnlabel{obs}
\end{align}
where, $\hat{\x}\in\Real^{\nx}$ denotes the estimate of the state vector, $\hat{\z}\in\Real^{\nz}$ denotes the estimate of $\z$, and the $\Lg \in \Real^{\nx \times \ny}$ is the unknown observer gain.
Let us define the state estimation error $\e$, and the observer error $\zerr$ as
\begin{align}
    \e := \x-\hat{\x}, \ \ \ \zerr := \z-\hat{\z}.  \eqnlabel{xzerr} 
\end{align}
Therefore, from \eqn{sys} and \eqn{obs},
\begin{align}
    \dot{\e} &= \left(\A+\Lg\Cy\right)\e + \BDelta\A\x + (\Bd+\BDelta\B)\dd  \nonumber \\ &\qquad\qquad + \Lg\Dd\dd + \Lg\Sn\n \eqnlabel{def_edot}
\end{align}
The observation error dynamics follows from \eqn{sys} and \eqn{def_edot} as
\begin{align}
    \dot{\tilde{\x}} = (\vt{A} + \BDelta\vt{A})\vt{x} + (\vt{B}\vt{S} + \BDelta\vt{B}) \vt{w}, \ \ \
    \zerr = \vt{C} \vt{x}, \eqnlabel{obs_err}
\end{align}
where $\vt{x}:=[\x^T,\e^T]^T$ is the augmented state vector, and $\vt{w}:=[\dd^T,\n^T]^T$ is the augmented vector of exogenous noises. And the augmented system matrices are as follows
\begin{equation}\left.
\begin{aligned}
  \vt{A} &:= \begin{bmatrix}\A & \vo{0} \\ \vo{0} & \A+\Lg\Cy \end{bmatrix}, \\
  \BDelta\vt{A} &:= \begin{bmatrix}\BDelta\A & \vo{0} \\ \BDelta\A & \vo{0} \end{bmatrix} = \vt{M}_1\F_1\vt{N}_1, \\
  \vt{B} &:= \begin{bmatrix}\Bd & \vo{0} \\ \Bd+\Lg\Dd & \Lg \end{bmatrix}, \vt{S}:= \begin{bmatrix}\I{}&\vo{0}\\ \vo{0} & \Sn \end{bmatrix}, \\
  \BDelta\vt{B} &:= \begin{bmatrix}\BDelta\Bd & \vo{0} \\ \BDelta\Bd & \vo{0} \end{bmatrix} = \vt{M}_2\F_2\vt{N}_2,  \\
  \vt{M}_i &:=[\M_i^T \ \ \M_i^T]^T, \; \vt{N}_i :=[\N_i \ \ \vo{0}] , i = 1,2.  
\end{aligned} \right\} \eqnlabel{def_tild_var}
\end{equation}
The objective is to determine the observer gain $\Lg$ such that the error system \eqn{obs_err} is stable, and the effect of $\vt{w}$ on $\zerr$ is bounded by the specified performance index. The sparse robust observer design problem will be stated in Section \ref{sec:sensor_prec}.
Now lets consider the systems with unstructured uncertainty.
\subsection{Systems with unstructured uncertainty}
\subsubsection{Plant}
Consider the uncertain plant
\begin{subequations}
\begin{align}
    \xdot &= \A(\voDelta)\x + \Bd(\voDelta)\dd, \\
    \y &= \Cy(\voDelta)\x + \Dd(\voDelta)\dd + \Dn\Sn\n, \\
    \z &= \Cz\x,
\end{align} \eqnlabel{sys2}
\end{subequations}
where the coefficient matrices are dependent on the uncertain parameters $\voDelta$. As in \eqn{sys}, here also we  assume that the process is independent of the sensor noise, and $\Dn=\I{}$. Systems such as \eqn{sys2} can be expressed  in linear fractional transformation (LFT) framework \cite{zhouBookRobust} as shown in  \fig{lft}.

The signals $\w_{\Delta}$ and $\z_{\Delta}$  are so-called fictitious input and output of the plant due to uncertainty block $\voDelta$. Let $\vo{\Delta}$ be uncertainty block such that $\norm{\vo{\Delta}}{\infty} \leq 1$.
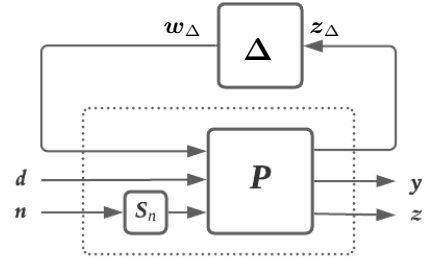
\begin{figure}[htb]
    \centering
    \includegraphics[trim=0.7cm 0.6cm 0.7cm 0.6cm,clip,width=0.32\textwidth]{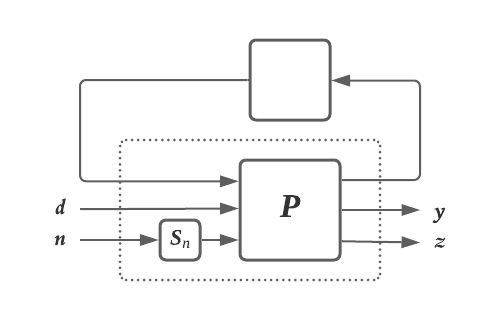}
    \begin{picture}(0,0)
      \put(-50,85){{\small $\z_{\Delta}$}}
      \put(-105,85){{\small $\w_{\Delta}$}}
      \put(-75,75){{\large $\voDelta$}}
    \end{picture}
    \caption{Plant with unstructured uncertainty.}
    \label{fig:lft}
\end{figure}
The system in \fig{lft} can be equivalently written the following state-space form
\begin{subequations}
  \begin{align}
    \xdot &= \A\x + \begin{bmatrix} \B_{\Delta} & \Bd & \vo{0} \end{bmatrix}\vt{w} ,\\
    \z_{\Delta} &= \C_{\Delta}\x +  \begin{bmatrix} \E_{\Delta} & \E_d & \vo{0} \end{bmatrix}\vt{w}, \\
    \y &= \Cy\x + \begin{bmatrix} \D_{\Delta} & \Dd & \Sn \end{bmatrix}\vt{w},\\
    \z &= \Cz\x,  \\
    \w_{\Delta} &= \vo{\Delta}\z_{\Delta},
  \end{align}  \eqnlabel{sys2_cut}
\end{subequations}
where $\vt{w}:= \begin{bmatrix} \w_{\Delta}^T &\dd^T & \n^T \end{bmatrix}^T$, and  $\A,\Bd$ etc. denote nominal values of the uncertain matrices $\A(\voDelta), \Bd(\voDelta)$. Top four equations in \eqn{sys2_cut} denote the state-space equations for the open loop system shown in the dotted box in \fig{lft}.

\subsubsection{Observer and error system}
We consider a state observer of the same form as \eqn{obs} for the system \eqn{sys2_cut}, with state estimation error and observer error as defined in \eqn{xzerr}. Therefore, error dynamics equations follow as
\begin{align}
    \dot{\e} = \left(\A+\Lg\Cy\right)\e + \vt{B}\vt{S} \vt{w}, \ \ \
    \zerr = \Cz\e, \eqnlabel{obs_err2}
\end{align}
\begin{subequations}
\begin{align}
  \text{where, } \vt{B} &:=  \begin{bmatrix} \B_{\Delta} & \Bd & \vo{0} \end{bmatrix}  + \Lg\begin{bmatrix} \D_{\Delta} & \Dd & \I{}\end{bmatrix}, \\
  \vt{S}&:= \diag\left( \I{}, \I{}, \Sn \right). 
\end{align}  \eqnlabel{def_tild_var2}
\end{subequations}
The objective here is to determine the observer gain $\Lg$ such that the transfer function from $\vt{w}$ to $[\z_{\Delta}^T, \zerr^T]^T$ is stable and bounded, and  the effect of $\vt{w}$ on $\zerr$ should be minimal.
 A formal problem statement will be presented in the following section.

\subsection{Sensor precision and observer design} \label{sec:sensor_prec}
As mentioned earlier, $\dd$ and $\n$ are $\mathcal{L}_2$-norm bounded but arbitrary signals.
Let us denote the sensor noise entering the system by $\vt{n}:=\Sn\n$, and let $\vt{n}_i$ denote the $i^{\text{th}}$ component of the signal $\vt{n}$. We define sensor precision to be the reciprocal of square of $\mathcal{L}_2$-norm (or \textit{energy}) of a signal, i.e. precision of the $i^{\text{th}}$ sensor channel is $1/\norm{\vt{n}_i}{2}^2$.

Since $\Sn>0$ is a diagonal matrix, let us define $0<\vo{\beta} =\left[\beta_1,\cdots, \beta_{\ny}\right]^T \in \Real^{\ny}$ such that
\begin{align}
   \diag{}(\vo{\beta}) := (\Sn\Sn)^{-1}. \eqnlabel{def_beta}
\end{align}
Therefore, precision of the $i^{\text{th}}$ sensor then becomes $1/\norm{\vt{n}_i}{2}^2 = \beta_i/\norm{\n_i}{2}^2$. Further, without loss of generality, we assume that $\norm{\n_i}{2}^2 = 1$. Therefore, precision of the $i^{\text{th}}$ sensor is simply $\beta_i$, and $\vo{\beta}$ is the precision vector.

Since $\vo{\beta}$ is interpreted as the precision vector, a sparse sensor configuration can be characterized by a sparse vector $\vo{\beta}$. If $\beta_i = 0$ for some sensor, it implies that the $i^{\text{th}}$ sensor noise channel contains infinite energy, or equivalently, that sensor  is not used.

Minimizing the number of non-zero elements in $\vo{\beta}$, i.e. $\norm{\vo{\beta}}{0}$, would yield a sparse configuration. However, minimization of $l_0$-norm is a non-convex problem, and in general very difficult to solve, especially for large-scale systems. A natural relaxation for the sparse configuration problem is minimization of $l_1$-norm instead. The minimization of $\norm{\vo{\beta}}{1}$ promotes sparsity, and as discussed in Section \ref{sec:iterate}, iterative reweighting techniques can be used to arrive at a sparse configuration. Such iterative techniques minimize weighted $l_1$-norm defined as
$$ \norm{\vo{\beta}}{1, \vo{\rho}}:= \vo{\rho}^T|\vo{\beta}|,$$
where $0<\vo{\rho}\in\Real^{\ny}$ is a specified weighting vector. Next, we formally define the sparse sensing problem for uncertain systems as follows.

\subsubsection{Observer design problem for system \eqn{sys}}

For the observer error system \eqn{obs_err}, let $\Gwz$ be the transfer function matrix from $\vt{w}$ to $\zerr$. We wish to minimize the effect of $\vt{w}$ on the observation error $\zerr$, which can be achieved by ensuring $\norm{\Gwz}{\infty} < \gamma$, i.e. $\norm{\zerr}{2}\leq\gamma\norm{\vt{w}}{2}$, for some $\gamma>0$. Therefore, the robust sparse sensing problem is formally stated as:
\begin{equation} 
  \begin{aligned}
    &\textit{Given uncertain system \eqn{sys} and error system \eqn{obs_err},} \\
    & \textit{given $\gamma>0$ and $\vo{\rho}>0$, determine optimal $\Lg,\vo{\beta}$} \\
    & \textit{that minimize } \norm{\vo{\beta}}{1,\vo{\rho}} \textit{ such that } \norm{\Gwz}{\infty} \leq \gamma \\
    & \textit{for all admissible } \BDelta\A \textit{ and } \BDelta\Bd.
     \eqnlabel{hinf_prob}
  \end{aligned} 
\end{equation}
\subsubsection{Observer design problem for the system in \eqn{sys2_cut}}
Consider the system in \eqn{sys2_cut} and the error system in \eqn{obs_err2}.
As mentioned earlier, the transfer function from $\vt{w}$ to $[\z_{\Delta}^T, \zerr^T]^T$ should be stable and bounded.
This is guaranteed if we ensure that the $\Hinf$ norm of the transfer function from $\vt{w}$ to $\zerr$ is less than a specified performance $\gamma>0$, i.e. $\norm{\Gwz}{\infty}\leq\gamma$, and for the transfer function from $\vt{w}$ to $\z_{\Delta}$ we require  $\norm{\vo{\mathcal{G}}_{\vt{w}\z_{\Delta}}}{\infty}\leq 1$ for the overall stability.
Therefore, the robust sparse sensing problem is then:
\begin{equation} 
  \begin{aligned}
    &\textit{Given the uncertain system in \eqn{sys2_cut}, and the error} \\
    & \textit{system in \eqn{obs_err2}, given $\gamma>0$ and $\vo{\rho}>0$, determine} \\
    & \textit{optimal $\Lg,\vo{\beta}$ that minimize } \norm{\vo{\beta}}{1,\vo{\rho}} \textit{ such that} \\
    & \norm{\Gwz}{\infty} \leq \gamma \textit{ and } \norm{\vo{\mathcal{G}}_{\vt{w}\z_{\Delta}}}{\infty}\leq 1 \\
    & \textit{ for all admissible } \voDelta.
     \eqnlabel{hinf_prob2}
  \end{aligned} 
\end{equation}

In the following section, we present the solutions to the problems defined in \eqn{hinf_prob} and \eqn{hinf_prob2}.

\section{Robust Sparse $\Hinf$ Observers} \label{sec:thms}
Before proceeding to the main results, we present the following lemmas, which will be useful in completing the proofs.
\subsection{Preliminaries}
\begin{lemma}[Schur complement \cite{lmiCSys}] \label{lem:schur}
Let $\X$ be a well-partitioned matrix defined as
 \begin{align*}
   \X:=\begin{bmatrix} \P & \Q \\ \Q^T & \R \end{bmatrix},
 \end{align*}
then, $\X<0$, if and only if, $\R<0$ and $\P - \Q\R^{-1}\Q^T < 0$.
\end{lemma}

\begin{lemma}[Variable elimination \cite{lmiCSys}] \label{lem:var_el}
Let $\Q,\X,\Y,\F$ be real matrices of appropriate dimensions. Then the following statements are equivalent.
\begin{itemize}
  \item[(a)] $\Q+\X\F\Y + \Y^T\F^T\X^T <0$ for all $\F$ that satisfy $\F^T\F\leq\I{}$.
  \item[(b)] There exists a scalar $\delta>0$ such that $\Q+\delta\X\X^T + \delta^{-1}\Y^T\Y <0.$
\end{itemize}
\end{lemma}

\begin{lemma}[Wang et al. \cite{wang1992SCL}] \label{lem:var_el2}
Let $\P,\Q,\X,\Y,\F$ be real matrices of appropriate dimensions such that $\P>0$ and $\F^T\F\leq\I{}$. If there exists a scalar $\delta>0$ such that $\P^{-1} - \delta^{-1}\X\X^T>0$, then the following is true.
\begin{align*}
  (\Q+\X\F\Y)^T & \P(\Q+\X\F\Y) \\ & \leq \Q^T(\P^{-1} - \delta^{-1}\X\X^T)^{-1}\Q + \delta\Y^T\Y.
\end{align*}
\end{lemma}

\begin{lemma}[Bounded real \cite{lmiCSys}] \label{lem:bounded}
Consider the LTI system,
\begin{align*}
    \xdot = \A\x + \B\w, \,\,
     \y = \C\x + \D\w,
\end{align*}
and its transfer function matrix $\vo{\mathcal{G}}(s) := \C(s\I{}-\A)^{-1}\B+\D$.
Then for a given scalar $\gamma>0$, $\norm{\vo{\mathcal{G}}(s)}{\infty}\leq\gamma$ if and only if there exists a matrix $\X>0$ such that
\begin{align*}
  \begin{bmatrix}
    \X\A + \A^T\X + \C^T\C &  \X\B +\C^T\D \\
    \ast & \D^T\D - \gamma^2\I{}
  \end{bmatrix} < 0.
\end{align*}
\end{lemma}

Proof of the following lemma can be easily established by starting with Lemma \ref{lem:bounded}, followed by manipulations using the results of Lemmas \ref{lem:schur}, \ref{lem:var_el} and \ref{lem:var_el2}.
\begin{lemma} \label{lem:bounded_rob}
Consider the LTI system,
\begin{align*}
    \xdot = (\A+\BDelta\A)\x + (\B+\BDelta\B)\w, \,\,
     \y = \C\x + \D\w,
\end{align*}
such that $\BDelta\A\in\set{A}$ and  $\BDelta\B\in\set{B}$ as defined in \eqn{def_unc_set}, and let
$\vo{\mathcal{G}}(s)$ be the associated transfer function matrix.
Then for a given scalar $\gamma>0$, $\norm{\vo{\mathcal{G}}(s)}{\infty}\leq\gamma$ for all admissible $\BDelta\A$ and $\BDelta\B$, if there exist a matrix $\X>0$, scalars $\delta_1>0$ and  $\delta_2>0$ such that
\begin{align*}
  \begin{bmatrix}
    \Z & \X\B +\C^T\D & \X\M_1        & \X\M_2 \\
    \ast    & \D^T\D - \gamma^2\I{} + \delta_2\N_2^T\N_2      & \vo{0}        & \vo{0} \\
    \ast    & \ast         & -\delta_1\I{} & \vo{0} \\
    \ast    & \ast         & \ast          & -\delta_2\I{}
  \end{bmatrix} < 0.
\end{align*}
where $\Z := \X\A + \A^T\X + \C^T\C + \delta_1\N_1^T\N_1$.
\end{lemma}
\subsection{Main result}
Next, we present the result for solving the $\Hinf$-optimal robust observer design problem \eqn{hinf_prob}.
\begin{theorem} \label{thm:hinf}
The optimal observer gain $\Lg$ and sensor precision $\vo{\beta}$ for the sparse $\Hinf$-optimal robust observer design problem \eqn{hinf_prob} is determined by solving the following optimization problem, and if the problem is feasible then the gain is recovered as $\Lg = \X_2^{-1}\Y$.
\begin{equation}\left.
\begin{aligned}
& \min\limits_{\X_1>0,\X_2>0,\Y,\vo{\beta}>0, \delta_1>0, \delta_2>0}\quad \norm{\vo{\beta}}{1,\vo{\rho}}  \\
 & \text{ such that }
  \begin{bmatrix}
    \Z_{11} & \Z_{12}  &\Z_{13}    & \Z_{14} \\
    \ast        &  \Z_{22}   & \vo{0}        & \vo{0} \\
    \ast       & \ast         & -\delta_1\I{} & \vo{0} \\
    \ast       & \ast         & \ast          & -\delta_2\I{}
  \end{bmatrix} < 0 ,
\end{aligned} \right\}\eqnlabel{hinf_thm}
\end{equation}
\begin{align*}
  & \Z_{11} = \\
  & \begin{bmatrix}\sym{\X_1\A} + \delta_1\N_1^T\N_1 &\vo{0}\\ \vo{0} & \sym{\X_2\A+\Y\Cy} + \Cz^T\Cz \end{bmatrix}\\
  &\Z_{12} = \begin{bmatrix}\X_1\Bd &\vo{0}\\  \X_2\Bd+\Y\Dd & \Y \end{bmatrix}, \\
  & \Z_{13} = \begin{bmatrix}\X_1\M_1\\ \X_2\M_1 \end{bmatrix} , \, \Z_{14} = \begin{bmatrix}\X_1\M_2\\ \X_2\M_2 \end{bmatrix}, \\
  & \Z_{22} = \begin{bmatrix} -\gamma^2\I{} + \delta_2\N_2^T\N_2 &\vo{0}\\ \vo{0} & -\gamma^2 \ \diag(\vo{\beta}) \end{bmatrix} .
\end{align*}
\end{theorem}
\begin{proof}
As a direct application of Lemma \ref{lem:bounded_rob} for the system \eqn{obs_err}, the condition $\norm{\Gwz}{\infty} \leq \gamma$ in \eqn{hinf_prob} is satisfied if there exist a symmetric matrix $\vt{X}>0$, $\delta_1>0$ and $\delta_2>0$ such that
\begin{align}
  \begin{bmatrix}
    \Z_{11} & \vt{X}\vt{B}\vt{S}  & \vt{X}\vt{M}_1    & \vt{X}\vt{M}_2 \\
    \ast        &  \W  & \vo{0}        & \vo{0} \\
    \ast       & \ast         & -\delta_1\I{} & \vo{0} \\
    \ast       & \ast         & \ast          & -\delta_2\I{}
  \end{bmatrix} < 0, \eqnlabel{hinf_tmp1}
\end{align}
\begin{align}
\text{where, } \Z_{11} &:=  \sym{\vt{X}\vt{A}}+ \vt{C}^T\vt{C} + \delta_1\vt{N}_1^T\vt{N}_1, \eqnlabel{hinf_tmpZ11} \\
\W &:= - \gamma^2\I{} + \delta_2\vt{N}_2^T\vt{N}_2. \nonumber
\end{align}

Using the result of Lemma \ref{lem:schur} successively, the inequality \eqn{hinf_tmp1} becomes
\begin{align}
  \begin{bmatrix}
    \Z_{11} & \Z_{12}  &\Z_{13}    & \Z_{14} \\
    \ast        &  \Z_{22}   & \vo{0}        & \vo{0} \\
    \ast       & \ast         & -\delta_1\I{} & \vo{0} \\
    \ast       & \ast         & \ast          & -\delta_2\I{}
  \end{bmatrix} < 0, \eqnlabel{hinf_tmp3}
\end{align}
where, $\Z_{12}:=\vt{X}\vt{B}$, $\Z_{13}:=\vt{X}\vt{M}_1$, $\Z_{14}:=\vt{X}\vt{M}_2$, and
\begin{align*}
\Z_{22} &:= \vt{S}^{-1}\W\vt{S}^{-1} \\
&= \diag\left((-\gamma^2\I{} + \delta_2\N_2^T\N_2) \ , \ -\gamma^2(\Sn\Sn)^{-1} \right)  \\ 
&= \diag\left((-\gamma^2\I{} + \delta_2\N_2^T\N_2) \ , \ -\gamma^2 \ \diag(\vo{\beta}) \right), 
\end{align*}
wherein we have used \eqn{def_beta} and the definitions of $\vt{N}_2$ and $\vt{S}$ from \eqn{def_tild_var}.
We partition $\vt{X}$ using $\X_1>0$, $\X_2>0$, such that
$\vt{X} := \diag\left(\X_1, \X_2 \right).$ 
Let us define $\Y:=\X_2\Lg$. Then $\Z_{11}$ follows from \eqn{hinf_tmpZ11} and \eqn{def_tild_var} as follows
\begin{align*}
& \Z_{11} = \\ & \begin{bmatrix}\sym{\X_1\A} + \delta_1\N_1^T\N_1 &\vo{0}\\ \vo{0} & \sym{\X_2\A+\Y\Cy} + \Cz^T\Cz \end{bmatrix}
\end{align*}
\begin{align*}
\text{Similarly, } \quad \Z_{12} &= \begin{bmatrix}\X_1\Bd &\vo{0}\\  \X_2\Bd+\Y\Dd & \Y \end{bmatrix}, \\
\Z_{13} &= \begin{bmatrix}\X_1\M_1\\ \X_2\M_1 \end{bmatrix} , \ \Z_{14} = \begin{bmatrix}\X_1\M_2\\ \X_2\M_2 \end{bmatrix}.
\end{align*}
Note that the inequality \eqn{hinf_tmp3} is linear in unknown variables $\X_1,\X_2,\Y,\vo{\beta}, \delta_1$ and $\delta_2$, and defines a feasibility condition for \eqn{hinf_prob}. The optimal solution is determined by minimizing the weighted $l_1$-norm $\norm{\vo{\beta}}{1,\vo{\rho}}$, and the observer gain can be recovered as $\Lg = \X_2^{-1}\Y$.
\end{proof}

The following theorem presents a solution to the problem \eqn{hinf_prob2}.
\begin{theorem} \label{thm:hinf2}
The optimal observer gain $\Lg$ and sensor precision $\vo{\beta}$ for the sparse $\Hinf$-optimal robust observer design problem \eqn{hinf_prob2} is determined by solving the following optimization problem, and if the problem is feasible then the gain is recovered as $\Lg = \X_2^{-1}\Y$.
\begin{equation}\left.
\begin{aligned}
& \min\limits_{\X_1>0,\X_2>0,\Y,\vo{\beta}>0} \ \norm{\vo{\beta}}{1,\vo{\rho}}  \text{ such that } \\
&   \begin{bmatrix}
      \Z_{11}  &  \Z_{12}   & \Z_{13}  \\
    \ast        &  \E_{\Delta}^T\E_{\Delta}-\I{}   & \E_{\Delta}^T\E_{d} \\
    \ast       & \ast         & \E_{d}^T\E_{d}-\I{}
  \end{bmatrix} < 0 ,\\
&   \begin{bmatrix}
    \W_{11} & \W_{12}  &\W_{13}    & \Y\\
    \ast        &  -\gamma^2\I{}   & \vo{0}        & \vo{0} \\
    \ast       & \ast         & -\gamma^2\I{} & \vo{0} \\
    \ast       & \ast         & \ast          & -\gamma^2 \ \diag(\vo{\beta})
  \end{bmatrix} < 0 ,
\end{aligned} \right\}\eqnlabel{hinf_thm2}
\end{equation}
\begin{align*}
& \Z_{11} = \sym{\X_1\A}+\C_{\Delta}^T\C_{\Delta}, \\
  & \Z_{12} =\X_1(\B_{\Delta}+\C_{\Delta}^T\E_{\Delta}),\
  \Z_{13} = \X_1(\Bd+\C_{\Delta}^T\E_d) ,\\
  & \W_{11} = \sym{\X_2\A+\Y\Cy} +\Cz^T\Cz, \\
  &\W_{12} = \X_2\B_{\Delta}+\Y\B_{\Delta}, \
   \W_{13} =\X_2\Bd+\Y\Bd.
\end{align*}
\end{theorem}
\begin{proof}
Consider the transfer function from $\vt{w}$ to $\z_{\Delta}$ from \eqn{sys2_cut}. The first LMI in the theorem statement follows directly from the Lemma \ref{lem:bounded} for the condition $\norm{\vo{\mathcal{G}}_{\vt{w}\z_{\Delta}}}{\infty}\leq 1$.

Now consider the condition $\norm{\Gwz}{\infty}\leq\gamma$ for the error system \eqn{obs_err2}. Using Lemma \ref{lem:bounded}, it becomes, for  $\X_2>0$
\begin{align*}
  \begin{bmatrix}
    \sym{\X_2(\A+\Lg\Cy)} +\Cz^T\Cz & \X_2\vt{B}\vt{S} \\
    \ast & -\gamma^2\I{}
  \end{bmatrix} < 0.
\end{align*}
Using the result of Lemma \ref{lem:schur} successively, we get
\begin{align*}
  \begin{bmatrix}
    \sym{\X_2(\A+\Lg\Cy)} +\Cz^T\Cz & \X_2\vt{B} \\
    \ast & -\gamma^2(\vt{S}\vt{S})^{-1}
  \end{bmatrix} < 0.
\end{align*}
Finally, by defining $\Y:=\X_2\Lg$, using the definitions of $\vt{B}$, $\vt{S}$ and $\vo{\beta}$ from \eqn{def_tild_var2} and \eqn{def_beta}, we arrive at the second LMI in the theorem statement, which concludes the proof.
\end{proof}

\subsection{Iterative reweighted $l_1$-minimization} \label{sec:iterate}
We use an iterative reweighting scheme presented in \cite{boyd2008weightedL1} to achieve a sparse sensor configuration. We perform multiple iterations of solving the optimization problems \eqn{hinf_thm} or \eqn{hinf_thm2}, and the weights for $(k+1)^{\text{th}}$ iteration are defined in terms of the previous iterate as
\begin{align}
  \rho^{(k+1)}_i = \left(\epsilon+|\beta_i^{(k)}|\right)^{-1}, \eqnlabel{reweight}
\end{align}
where $\epsilon>0$ is a small number which ensures that the weights are well defined. Initial weights are chosen to be equal, i.e. without loss of generality, $\rho^{(0)}_i = 1$.
These iterations are stopped if the convergence criterion is met or the maximum number of iterations is reached \cite{boyd2008weightedL1}. 

The final refined solution is determined by removing sensors with very small precisions and re-solving  \eqn{hinf_thm} or \eqn{hinf_thm2} with equal weights $\rho_i=1$.

\section{Example} \label{sec:ex}
Let us consider a serially connected spring-mass-damper system on a frictionless surface as shown in \fig{SMD}.
\begin{figure}[htb]
    \centering
    \includegraphics[trim=0.35cm 0.35cm 0.25cm 0.35cm,clip,width=0.4\textwidth]{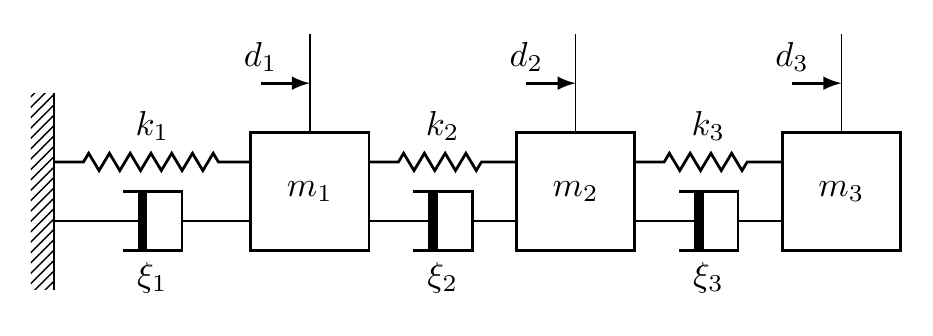}
    \caption{Serially connected spring-mass-damper system.}
    \label{fig:SMD}
\end{figure}

Let $x_i$ denote the distance of the $i^{\text{th}}$ mass from the wall. Define state vector to be $\x:=[x_1,x_2,x_3,\dot{x}_1,\dot{x}_2,\dot{x}_3]^T$. The nominal values of masses $m_i$, spring constants $k_i$, and damper coefficients $\xi_i$ are all assumed to be unity. Disturbances $d_i$ enters the system in the form of external forces acting independently on all masses. Also, we assume that sensors measure position and velocity of each mass. Therefore, there are six sensors. The first three sensors measure positions, and the last three sensors measure velocities of the masses.
The nominal system  matrices are given by
\begin{align*}
  \A &=  \begin{bmatrix} \vo{0} & \I{}\\ \H & \H \end{bmatrix} ,
  \H =  \begin{bmatrix} -2  &   1  &   0 \\    1  &  -2   &  1 \\    0    & 1   & -1\end{bmatrix},
  \Bd =  \begin{bmatrix} \vo{0} \\ \I{}\end{bmatrix}\Sd, \\
  \Cy &= \I{}, \Dd = \vo{0}, \Dn = \I{}, \Cz = \I{},
\end{align*}
where $\Sd$ is a known matrix which represents a scaling for the disturbance signal.
Next we consider the uncertainty of the form \eqn{sys} and \eqn{sys2_cut}, and determine the robust sparse observers using results of Theorems \ref{thm:hinf} and \ref{thm:hinf2}.

\subsubsection*{Uncertainty of the form \eqn{sys}}
The uncertainty in the system matrices is assumed to of the following form
\begin{align*}
  \BDelta\A = \begin{bmatrix} \vo{0} \\ \I{}  \end{bmatrix} \F_1 \begin{bmatrix} c_0 \H & \vo{0} \\ \vo{0} & c_1\H   \end{bmatrix},    \BDelta\Bd = \begin{bmatrix} \vo{0} \\ \I{}  \end{bmatrix} \F_2 (c_2\I{}).
\end{align*}
where  $\F_1^T\F_1\leq\I{}$, $\F_2^T\F_2\leq\I{}$ and $c_0,c_1,c_2$ are known non-negative constants which quantify the magnitude of uncertainty in the system, i.e., larger values of these parameters would imply uncertainty of larger magnitude, and on the other extreme end, zero-valued parameters correspond to the nominal plant with no uncertainty. For the system defined as above, we can directly apply the result of Theorem \ref{thm:hinf} with iterative reweighting \eqn{reweight} to determine sparse sensor configuration and corresponding optimal precision $\vo{\beta}$, which is discussed next.

The optimization problems are solved using the solver \texttt{SDPT3}\cite{SDPT3} with \texttt{CVX} \cite{cvx} as a parser.
 In the figures of this section, solid circle indicates that a sensor is required and the number above it shows the required precision $\vo{\beta}$, while cross indicates that a sensor is not used.

First, we consider the effect of specified performance parameter $\gamma$, for uncertainty of a fixed magnitude quantified by the parameters $c_0,c_1,c_2$. In particular, in \fig{varyWithGam}, we set $c_0,c_1,c_2$ to non-zero values, and vary the specified performance parameter $\gamma$. We observer that as we decrease $\gamma$, i.e. demand better performance, the number of required sensors and their associated precisions increase. For $\gamma =1$, only two sensors are needed, whereas for $\gamma = 0.25$, all six sensors are needed  with relatively higher precisions.
\begin{figure}[htb]
    \centering
    \includegraphics[trim=0.25cm 0.15cm 0.2cm 0.02cm,clip,width=0.48\textwidth]{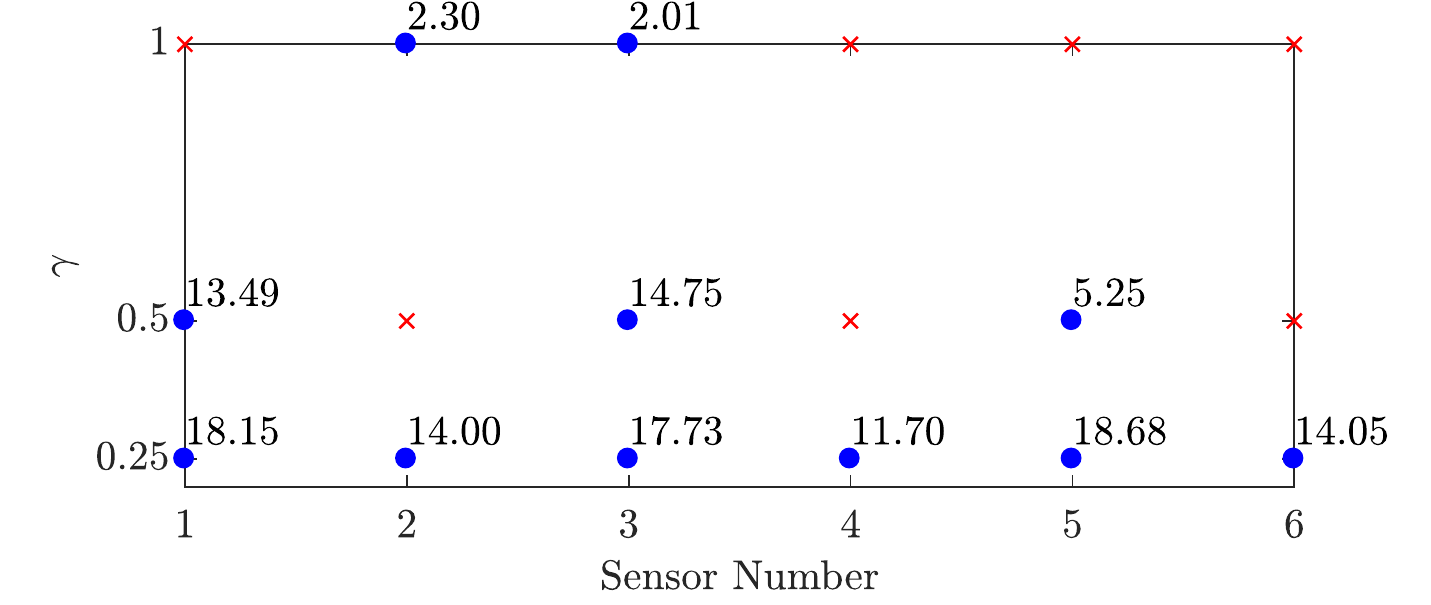}
    \caption{Sensor configuration and precision for different values of specified $\gamma$, and fixed $c_0 = 0.01$, $c_1 = 0.02$, $c_2 = 0.03$, and $\Sd = \I{}$.}
    \label{fig:varyWithGam}
\end{figure}
\begin{figure}[htb]
    \centering
    \includegraphics[trim=0.25cm 0.15cm 0.2cm 0.02cm,clip,width=0.48\textwidth]{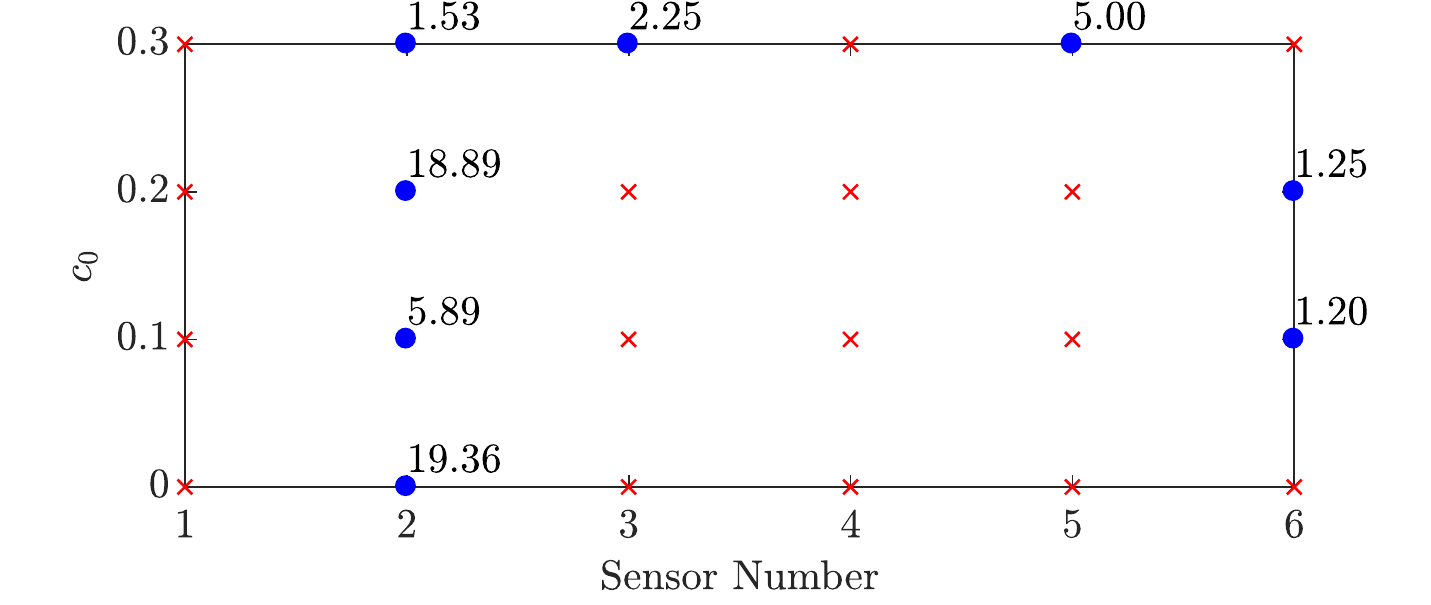}
    \caption{Sensor configuration and precision for different magnitudes of uncertainty quantified by $c_0$, and fixed $c_1 = 0$, $c_2 = 0$, $\gamma=1$, and $\Sd = \I{}$.}
    \label{fig:varyWithC0}
\end{figure}

In \fig{varyWithC0}, we show a complementary case to \fig{varyWithGam}, i.e. we vary the magnitude of uncertainty for a fixed performance parameter $\gamma$.
In particular, we vary $c_0$, for fixed parameters $c_1=0,c_2=0$, and $\gamma = 1$. As we increase the magnitude of uncertainty, i.e $c_0$, we see that the number of sensors required and the precision values increase. For the nominal plant corresponding to $c_0 = 0$, only one sensor is required, whereas for $c_0 = 0.3$, three sensors are required.
\subsubsection*{Uncertainty of the form \eqn{sys2_cut}}
We assume that the spring constants and damper coefficients can take values in the intervals $[c_0-1, c_0+1]$ and $[c_1-1, c_1+1]$ respectively, where $ c_0, c_1$ are known non-negative constants which quantify the magnitude of uncertainty. We assume no uncertainty in the masses. With such uncertainty, the system can be written in an LFT form as in \fig{lft} and \eqn{sys2_cut} such that $\norm{\voDelta}{\infty}\leq 1$.

Analogous to \fig{varyWithGam} and \fig{varyWithC0}, \fig{varyWithGam2} and \fig{varyWithC02} below  respectively show the effect of decreasing performance parameter $\gamma$ for a fixed magnitude of uncertainty, and the effect of increasing magnitude of uncertainty for a fixed  performance parameter $\gamma$. Similar observations as in the previous case can also be made for \fig{varyWithGam2} and \fig{varyWithC02}.
\begin{figure}[htb]
    \centering
    \includegraphics[trim=0.25cm 0.15cm 0.2cm 0.02cm,clip,width=0.48\textwidth]{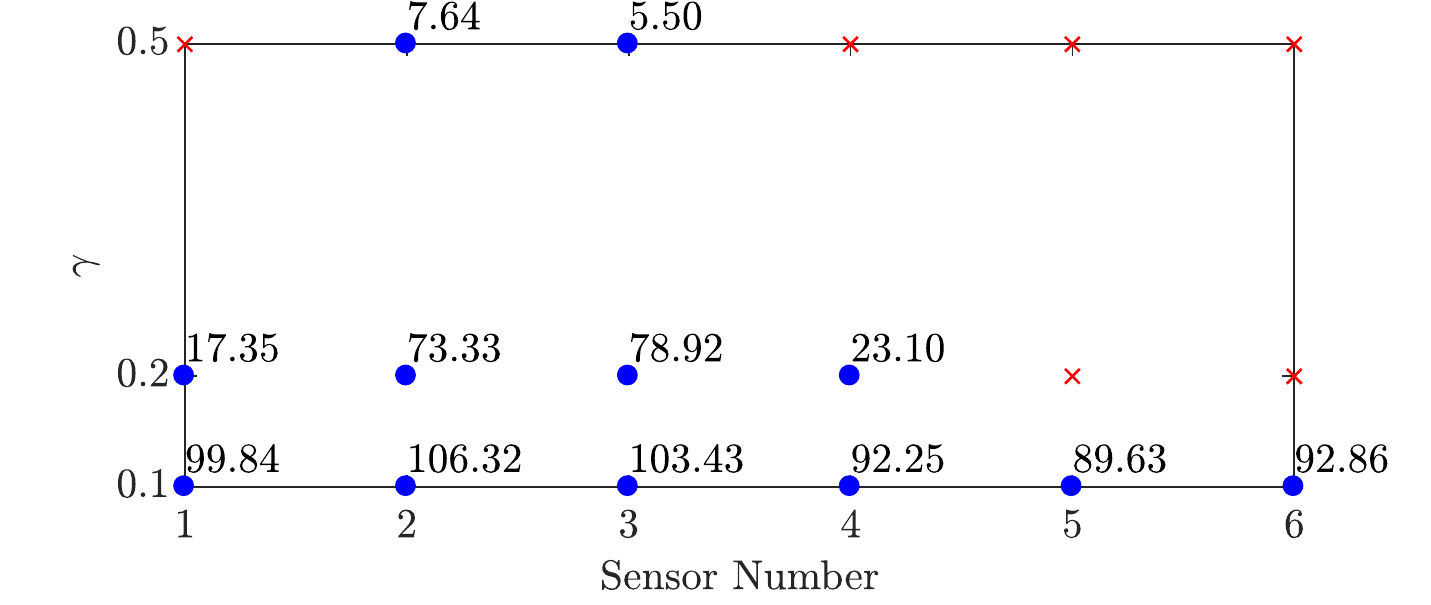}
    \caption{Sensor configuration and precision for different values of specified $\gamma$, and fixed $c_0 = 0.1$, $c_1 = 0.1$, $\Sd = 0.2\I{}$.}
    \label{fig:varyWithGam2}
\end{figure}

\begin{figure}[htb]
    \centering
    \includegraphics[trim=0.25cm 0.15cm 0.2cm 0.02cm,clip,width=0.48\textwidth]{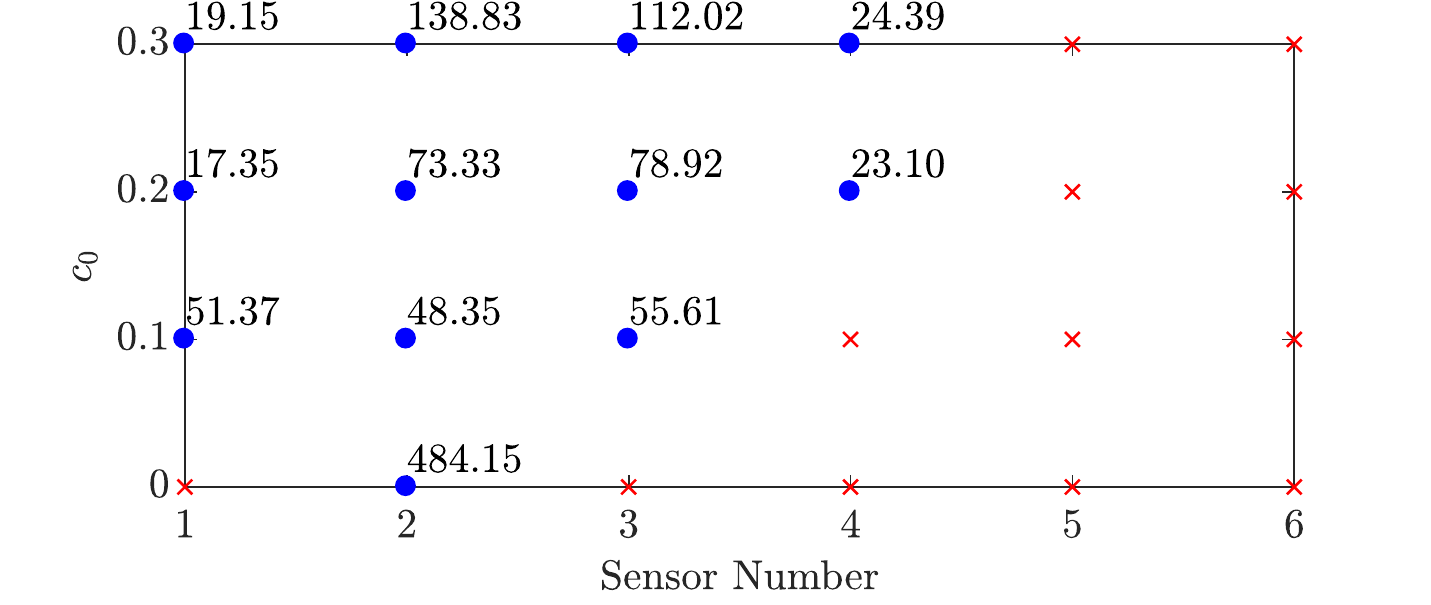}
    \caption{Sensor configuration and precision for different magnitudes of uncertainty quantified by $c_0$, and fixed $c_1 = 0$, $\gamma=0.2$, and $\Sd = 0.2\I{}$.}
    \label{fig:varyWithC02}
\end{figure}
\section{Conclusion} \label{sec:concl}
Herein, we present a unified theoretical framework to address the problem of sparse sensor configuration in the presence of model uncertainty while simultaneously minimizing the required sensor precisions. We consider two types of model uncertainties: structured affine uncertainty in system matrices and unstructured uncertainty. For both the cases, the robust sparse sensing problem is formulated in the context of $\Hinf$-optimal observer design and posed as a convex optimization problem subject to linear matrix inequalities. We minimize $l_1$-norm of the precision vector to promote sparsity, and an iterative reweighting scheme is used to refine the solution.

The convex optimization problems formulated in this work are semi-definite programs (SDPs). General-purpose solvers that we used in numerical simulations to solve the SDPs, in general, do not scale well for large-scale systems. These SDPs require customized solvers that can exploit the optimization problem's local structure, e.g., \cite{Jovanovic2019TAC}. However, for discussion brevity, this paper is limited to only the theoretical development of the framework. The development of customized algorithms for large-scale systems is a topic of our ongoing research.



\bibliographystyle{unsrt}
\bibliography{root}

\begin{thebibliography}{10}

\bibitem{boyd2009tac}
S.~{Joshi} and S.~{Boyd}.
\newblock Sensor selection via convex optimization.
\newblock {\em IEEE Transactions on Signal Processing}, 57(2):451--462, Feb
  2009.

\bibitem{Sundaram2017automatica}
H.~{Zhang}, R.~{Ayoub}, and S.~{Sundaram}.
\newblock Sensor selection for kalman filtering of linear dynamical systems:
  Complexity, limitations and greedy algorithms.
\newblock {\em Automatica}, 78:202--210, 2017.

\bibitem{pappas2016acc}
V.~{Tzoumas}, A.~{Jadbabaie}, and G.~J. {Pappas}.
\newblock Sensor placement for optimal kalman filtering: Fundamental limits,
  submodularity, and algorithms.
\newblock In {\em 2016 American Control Conference (ACC)}, pages 191--196, July
  2016.

\bibitem{das2017icssa}
N.~Das and R.~Bhattacharya.
\newblock Sparse sensing architecture for kalman filtering with guaranteed
  error bound.
\newblock In {\em 1st IAA ICSSA}, 2017.

\bibitem{skelton2008jour}
F.~Li, M.~C. de~Oliveira, and R.~Skelton.
\newblock Integrating information architecture and control or estimation
  design.
\newblock {\em SICE JCMSI}, 1(2):120--128, 2008.

\bibitem{saraf2017acc}
R.~Saraf, R.~Bhattacharya, and R.~Skelton.
\newblock H2 optimal sensing architecture with model uncertainty.
\newblock In {\em 2017 American Control Conference}, pages 2429--2434, 2017.

\bibitem{deshpande2021cdcLCSS}
V.~M. {Deshpande} and R.~{Bhattacharya}.
\newblock Sparse sensing and optimal precision: An integrated framework for
  $\mathcal{H}_2/\mathcal{H}_{\infty}$ optimal observer design.
\newblock {\em IEEE Control Systems Letters}, 5(2):481--486, 2021.

\bibitem{Lopez2014acc}
J.~{Lopez}, Y.~{Wang}, and M.~{Sznaier}.
\newblock Sparse $\mathcal{H}_ 2$ optimal filter design via convex
  optimization.
\newblock In {\em 2014 ACC}, pages 1108--1113, June 2014.

\bibitem{Wolfrum2014tac}
U.~{Münz}, M.~{Pfister}, and P.~{Wolfrum}.
\newblock Sensor and actuator placement for linear systems based on $h_{2}$ and
  $h_{\infty}$ optimization.
\newblock {\em IEEE Transactions on Automatic Control}, 59(11):2984--2989,
  2014.

\bibitem{Stefanopoulou2020battery}
X.~{Lin}, H.~E. {Perez}, J.~B. {Siegel}, and A.~G. {Stefanopoulou}.
\newblock Robust estimation of battery system temperature distribution under
  sparse sensing and uncertainty.
\newblock {\em IEEE Transactions on Control Systems Technology},
  28(3):753--765, 2020.

\bibitem{Jovanovic2019TAC}
A.~{Zare}, H.~{Mohammadi}, N.~K. {Dhingra}, T.~T. {Georgiou}, and M.~R.
  {Jovanovic}.
\newblock Proximal algorithms for large-scale statistical modeling and
  sensor/actuator selection.
\newblock {\em IEEE T AUTOMAT CONTR}, 2019.

\bibitem{Chepuri2015tac}
S.~P. {Chepuri} and G.~{Leus}.
\newblock Sparsity-promoting sensor selection for non-linear measurement
  models.
\newblock {\em IEEE Transactions on Signal Processing}, 63(3):684--698, 2015.

\bibitem{das2020ifac}
N.~Das and R.~Bhattacharya.
\newblock Optimal sensing precision in ensemble and unscented kalman filtering.
\newblock In {\em 21st IFAC World Congress}, 2020.
\newblock To appear. Preprint: arXiv:2003.06003.

\bibitem{hiramoto2000vibration}
K.~{Hiramoto}, H.~{Doki}, and B.~{Obinata}.
\newblock Optimal sensor/actuator placement for active vibration control using
  explicit solution of algebraic riccati equation.
\newblock {\em J. Sound Vibr.}, 229(5):1057--1075, 2000.

\bibitem{zhouBookRobust}
K.~Zhou, J.~C. Doyle, and K.~Glover.
\newblock {\em Robust and Optimal Control}.
\newblock Prentice-Hall, Inc., Upper Saddle River, New Jersey, 1 edition, 1996.

\bibitem{lmiCSys}
Guang-Ren Duan and Hai-Hua Yu.
\newblock {\em LMIs in Control Systems}.
\newblock CRC Press, Boca Raton, FL, 1 edition, 2013.

\bibitem{wang1992SCL}
Youyi Wang, Lihua Xie, and Carlos~E. de~Souza.
\newblock Robust control of a class of uncertain nonlinear systems.
\newblock {\em Systems \& Control Letters}, 19(2):139--149, 1992.

\bibitem{boyd2008weightedL1}
E.~J. {Candes}, M.~B. {Wakin}, and S.~P. {Boyd}.
\newblock Enhancing sparsity by reweighted $l_1$ minimization.
\newblock {\em J. Fourier Anal. Appl.}, 14:877–905, 2008.

\bibitem{SDPT3}
K.~C. Toh, M.~J.Todd, and R.~H. Tütüncü.
\newblock Sdpt3 — a matlab software package for semidefinite programming,
  version 1.3.
\newblock {\em Optimization Methods and Software}, 11(1-4):545--581, 1999.

\bibitem{cvx}
Michael Grant and Stephen Boyd.
\newblock {CVX}: Matlab software for disciplined convex programming, version
  2.1, March 2014.

\end{thebibliography}
\end{document}